\documentclass[journal]{IEEEtran}
\usepackage{bbm}
\usepackage{array}
\usepackage{dcolumn}
\usepackage{epsfig}
\usepackage[intlimits]{amsmath}
\usepackage{amsmath, amsfonts, yhmath, bm}
\usepackage{amssymb}
\usepackage{psfrag}
\usepackage{color,soul}
\usepackage[dvipsnames]{xcolor}
\usepackage[normalem]{ulem}
\usepackage{enumerate}
\usepackage{stackengine}
\usepackage[noadjust]{cite}
\usepackage{graphicx}
\usepackage{multirow}
\usepackage[caption=false,font=footnotesize]{subfig}
\usepackage{etoolbox}
\usepackage{tcolorbox}
\usepackage{float}
\usepackage{dsfont}
\usepackage{tikz}
\usetikzlibrary{arrows}
\usetikzlibrary{calc}
\usetikzlibrary{decorations.pathmorphing}
\usetikzlibrary{arrows.meta,positioning,fit}
\usepackage{svg}
\usepackage{comment}
\usepackage{ifthen}
\let\CMmathcal\mathcal      
\usepackage{eucal}         
\let\mathcal\CMmathcal     

\usepackage{hyperref}
\hypersetup{
    colorlinks,
    linkcolor={red!50!black},
    citecolor={blue!50!black},
    urlcolor={blue!80!black}
}

\usepackage{balance}
\usepackage{url}
\usepackage[inline]{enumitem} 

\usetikzlibrary{backgrounds, positioning, fit}
\usepackage{pgfplots}
\pgfplotsset{compat=newest}  

\usepackage{mathtools} 

\usepackage{vmr-symbols-vecbold}
\usepackage{standard-macros}
\PassOptionsToPackage{hyphens}{url}

\newcommand{\defeq}{\mathrel{\mkern-0.25mu=}}
\DeclareSymbolFontAlphabet{\amsmathbb}{AMSb}%

\newcommand{\lro}[1]{\lefto({#1}\right)}																

\newcommand{\lr}[1]{\left({#1}\right)}																
\newcommand{\lrb}[1]{\left \lbrace {#1} \right \rbrace}															

\safemath{\dopplerspread}{B_D}																								
\safemath{\delayspread}{T_D}																									
\safemath{\nc}{n\sub{c}}																										
\safemath{\nf}{n\sub{f}}																										
\safemath{\efa}{p\sub{sc}}
\safemath{\efb}{p\sub{cs}}
\safemath{\ef}{\epsilon\sub{f}	}
\safemath{\nd}{n\sub{d}}																										
\safemath{\ntx}{n\sub{t}} 																											
\safemath{\nrx}{n\sub{r}}																											
\safemath{\ntxt}{\tilde{n\sub{t}}}																											
\safemath{\cb}{\ensuremath{L}} 																								
\safemath{\cl}{\ensuremath{n}} 																								
\safemath{\txanto}{{\ensuremath{\tilde{m}_t}}} 																		
\safemath{\cs}{M} 																														
\safemath{\idPustm}{\ensuremath{S_{k}}}
\safemath{\error}{\ensuremath{\epsilon}} 																				
\safemath{\eexp}{\ensuremath{\mathcal{E}}} 																			
\safemath{\nsubc}{n\sub{s}}			 																						
\safemath{\nofdm}{n\sub{o}} 																									
\safemath{\bc}{\ensuremath{B_c}} 																							
\safemath{\ts}{\ensuremath{T_s}} 																							
\safemath{\nrb}{\ensuremath{n_{rb}}} 																						
\safemath{\rul}{\ensuremath{\rho\sub{ul}}}
\safemath{\rdl}{\ensuremath{\rho\sub{dl}}}

\safemath{\nres}{\ell}
\safemath{\nr}{n\sub{r}}
\safemath{\maxk}{M^*\lr{\nres, \nsubc, \nofdm, \epsilon, \rho}}
\safemath{\Rmax}{R^*}
\safemath{\Emin}{E\sub{b}^*/N_0}
\safemath{\Eminf}{\frac{E\sub{b}^*}{N_0}}
\safemath{\np}{\ensuremath{n\sub{p}}}
\safemath{\ndf}{\ensuremath{\bar{n}\sub{d}}}
\safemath{\npf}{\ensuremath{\bar{n}\sub{p}}}
\safemath{\code}{\ensuremath{\mathcal{C}}}
\safemath{\err}{\ensuremath{\epsilon}}
\safemath{\rp}{\ensuremath{\rho\sub{p}}}
\safemath{\rd}{\ensuremath{\rho\sub{d}}}
\safemath{\cohtime}{\ensuremath{T\sub{c}}}
\safemath{\cohbw}{\ensuremath{B\sub{c}}}
\safemath{\nmax}{\ensuremath{\ell\sub{m}}}
\safemath{\ntot}{\ensuremath{n\sub{tot}}}
\safemath{\nul}{\ensuremath{n\sub{ul}}}
\safemath{\ndl}{\ensuremath{n\sub{dl}}}

\safemath{\yp}{\ensuremath{\randvecy_{\nu}^{(\text{p})}}}
\safemath{\yd}{\ensuremath{\randvecy_{\nu}^{(\text{d})}}}
\safemath{\ypd}{\ensuremath{\vecy_{\nu}^{(\text{p})}}}
\safemath{\ydd}{\ensuremath{\vecy_{\nu}^{(\text{d})}}}

\safemath{\ypf}{\ensuremath{\bar{\randvecy}_{\nu}^{(\text{p})}}}
\safemath{\ydf}{\ensuremath{\bar{\randvecy}_{\nu}^{(\text{d})}}}
\safemath{\ypdf}{\ensuremath{\bar{\vecy}_{\nu}^{(\text{p})}}}
\safemath{\yddf}{\ensuremath{\bar{\vecy}_{\nu}^{(\text{d})}}}

\safemath{\xp}{\ensuremath{\vecx^{(\text{p})}}}
\safemath{\xd}{\ensuremath{\randvecx_{\nu}^{(\text{d})}}}
\safemath{\xdd}{\ensuremath{\vecx_{\nu}^{(\text{d})}}}

\safemath{\xpf}{\ensuremath{\bar{\vecx}^{(\text{p})}}}
\safemath{\xdf}{\ensuremath{\bar{\randvecx}_{\nu}^{(\text{d})}}}
\safemath{\xddf}{\ensuremath{\bar{\vecx}_{\nu}^{(\text{d})}}}

\safemath{\xdb}{\ensuremath{\overline{\randvecx}^{(\text{d})}}}
\safemath{\Pxd}{\ensuremath{P_{\randvecx^{(\text{d})}}}}

\safemath{\xpbar}{\ensuremath{\overline{\matX}^{(\text{p})}}}
\safemath{\xdbar}{\ensuremath{\overline{\randmatX}^{(\text{d})}}}

\safemath{\xdv}{\ensuremath{\randvecx^{(\text{d})}}}
\safemath{\xdbarv}{\ensuremath{\overline{\randvecx}^{(\text{d})}}}
\safemath{\ydv}{\ensuremath{\randvecy^{(\text{d})}}}

\safemath{\xdr}{\ensuremath{\matX^{(\text{d})}}}

\safemath{\ttx}{\ensuremath{\tau\sub{tx}}}
\safemath{\trx}{\ensuremath{\tau\sub{rx}}}
\safemath{\ack}{\ensuremath{\mathrm{s}}}
\safemath{\nack}{\ensuremath{\mathrm{c}}}

\safemath{\mI}{\ensuremath{i\lro{\randvecy ; \randvecx}}} 				

\safemath{\randveca}{\bm{A}}
\safemath{\randvecb}{\bm{B}}
\safemath{\randvecc}{\bm{C}}
\safemath{\randvecd}{\bm{D}}
\safemath{\randvece}{\bm{E}}
\safemath{\randvecf}{\bm{F}}
\safemath{\randvecg}{\bm{G}}
\safemath{\randvech}{\bm{H}}
\safemath{\randveci}{\bm{I}}
\safemath{\randvecj}{\bm{J}}
\safemath{\randveck}{\bm{K}}
\safemath{\randvecl}{\bm{L}}
\safemath{\randvecm}{\bm{M}}
\safemath{\randvecn}{\bm{N}}
\safemath{\randveco}{\bm{O}}
\safemath{\randvecp}{\bm{P}}
\safemath{\randvecq}{\bm{Q}}
\safemath{\randvecr}{\bm{R}}
\safemath{\randvecs}{\bm{S}}
\safemath{\randvect}{\bm{T}}
\safemath{\randvecu}{\bm{U}}
\safemath{\randvecv}{\bm{V}}
\safemath{\randvecw}{\bm{W}}
\safemath{\randvecx}{\bm{X}}
\safemath{\randvecy}{\bm{Y}}
\safemath{\randvecz}{\bm{Z}}
\safemath{\randvecphi}{\bm{\Phi}}

\safemath{\randmatA}{\amsmathbb{A}}
\safemath{\randmatB}{\amsmathbb{B}}
\safemath{\randmatC}{\amsmathbb{C}}
\safemath{\randmatD}{\amsmathbb{D}}
\safemath{\randmatE}{\amsmathbb{E}}
\safemath{\randmatF}{\amsmathbb{F}}
\safemath{\randmatG}{\amsmathbb{G}}
\safemath{\randmatH}{\amsmathbb{H}}
\safemath{\randmatI}{\amsmathbb{I}}
\safemath{\randmatJ}{\amsmathbb{J}}
\safemath{\randmatK}{\amsmathbb{K}}
\safemath{\randmatL}{\amsmathbb{L}}
\safemath{\randmatM}{\amsmathbb{M}}
\safemath{\randmatN}{\amsmathbb{N}}
\safemath{\randmatO}{\amsmathbb{O}}
\safemath{\randmatP}{\amsmathbb{P}}
\safemath{\randmatQ}{\amsmathbb{Q}}
\safemath{\randmatR}{\amsmathbb{R}}
\safemath{\randmatS}{\amsmathbb{S}}
\safemath{\randmatT}{\amsmathbb{T}}
\safemath{\randmatU}{\amsmathbb{U}}
\safemath{\randmatV}{\amsmathbb{V}}
\safemath{\randmatW}{\amsmathbb{W}}
\safemath{\randmatX}{\amsmathbb{X}}
\safemath{\randmatY}{\amsmathbb{Y}}
\safemath{\randmatZ}{\amsmathbb{Z}}
\safemath{\randmatSigma}{\mathbb{\Sigma}}
\safemath{\randmatPhi}{\mathbb{\Phi}}
\safemath{\randmatLambda}{\mathbb{\Lambda}}

\safemath{\matSigma}{\bm{\Sigma}}
\safemath{\matPhi}{\bm{\Phi}}
\safemath{\matLambda}{\bm{\Lambda}}

\newcommand{\ratedist}{\text{\tiny RD}}
\newcommand{\selscheme}{\text{\tiny sel}}

\newcommand{\addra}{}
\newcommand{\delra}[1]{}

\definecolor{mitgrayback}{cmyk}{0.06, 0.05, 0.05, 0.0}
\definecolor{lightblack}{cmyk}{0.3, 0.3, 0.3, 0.25}
\definecolor{lightgray}{cmyk}{0.02, 0, 0, 0.0}

\makeatletter
\newsavebox\myboxA
\newsavebox\myboxB
\newlength\mylenA

\newcommand*\mybar[2][0.75]{%
	\sbox{\myboxA}{$\m@th#2$}%
	\setbox\myboxB\null
	\ht\myboxB=\ht\myboxA%
	\dp\myboxB=\dp\myboxA%
	\wd\myboxB=#1\wd\myboxA
	\sbox\myboxB{$\m@th\overline{\copy\myboxB}$}
	\setlength\mylenA{\the\wd\myboxA}
	\addtolength\mylenA{-\the\wd\myboxB}%
	\ifdim\wd\myboxB<\wd\myboxA%
	\rlap{\hskip 0.5\mylenA\usebox\myboxB}{\usebox\myboxA}%
	\else
	\hskip -0.5\mylenA\rlap{\usebox\myboxA}{\hskip 0.5\mylenA\usebox\myboxB}%
	\fi}
\makeatother

\makeatletter
\renewcommand*{\defeq}{\mathrel{\rlap{%
			\raisebox{0.3ex}{$\m@th\cdot$}}%
		\raisebox{-0.3ex}{$\m@th\cdot$}}%
	=}
\makeatother

\newtheorem{remark}{Remark}
\newtheorem{theorem}{Theorem}

\newtheorem{lemma}{Lemma}
\begin{document}

\title{On the Suboptimality of Rate--Distortion-Optimal Compression: Fundamental Accuracy Limits for Distributed Localization}

\author{Amir~Weiss
\thanks{A. Weiss is with The Alexander Kofkin Faculty of Engineering, Bar-Ilan University, Ramat Gan, 5290002, Israel, e-mail: amir.weiss@biu.ac.il.}}

\maketitle

\begin{abstract}
We derive fundamental accuracy limits for distributed localization when a fusion center \delra{observes}\addra{has access} only \addra{to independently }rate--distortion (RD)-optimally compressed versions of \delra{wideband }multi-sensor \delra{waveforms}\addra{observations, under a line-of-sight propagation model with a Gaussian wideband waveform}. Using the Gaussian RD test-channel model together with a Whittle spectral Fisher-information characterization, we obtain an explicit frequency-domain Cram\'er--Rao lower bound. A two-band, two-level specialization yields closed-form expressions and reveals a rate-induced regime change: RD-optimal compression under a squared-error distortion measure can eliminate localization-informative spectral content. A simple band-selective scheme can outperform RD compression by orders of magnitude at the same rate, motivating localization-aware compression for networked sensing and integrated sensing and communication systems.
\end{abstract}

\vspace{-0.1cm}
\begin{IEEEkeywords}
Distributed localization, rate--distortion theory, task-oriented compression, Cram\'er--Rao lower bound.
\end{IEEEkeywords}

\IEEEpeerreviewmaketitle
\vspace{-0.45cm}
\section{Introduction}\label{sec:introduction}
The proliferation of distributed sensing and communication platforms---including Internet-of-Things deployments~\cite{da2014internet}, \delra{battery-powered }sensor networks~\cite{wu2023low}, wearable devices~\cite{sztyler2016body}, and edge-assisted cyber--physical systems~\cite{lu2022edge}---has accelerated the need to operate under stringent resource constraints. In such settings, raw waveform streaming from multiple nodes to a fusion center (FC) is often infeasible due to limited uplink rates, energy budgets, latency requirements, or shared-spectrum constraints. As a result, modern architectures increasingly rely on \emph{compressed observations}~\cite{ketshabetswe2021data}, where each node transmits a rate-constrained representation of its measurements\delra{ (e.g., quantized or otherwise compressed}~\cite{Shlezinger2019TaskQuant}\delra{)}, while inference is performed centrally or cooperatively. Understanding the fundamental impact of compression on estimation fidelity~\cite{zhang1988estimation} is therefore critical for principled system design.

A particularly important inference task is \emph{distributed localization}, which underpins numerous applications such as \delra{asset tracking, }navigation\delra{ in GNSS-denied environments, autonomous driving}, robotics, and industrial monitoring~\cite{trevlakis2023localization}. \delra{Moreover, localization}\addra{It} has\addra{ also} emerged as a key capability in \delra{the context of }integrated sensing and communications (ISAC)~\cite{liu2022integrated,liu2022survey,lu2024integrated}, where \delra{communication waveforms and infrastructures are leveraged to sense the environment and localize targets or users. In many ISAC and networked sensing scenarios, localization must be performed from observations collected across }distributed receivers\addra{ collect observations that must often be}\delra{ and} relayed over rate-limited links\delra{, making}\addra{. This makes} the interplay between compression and localization accuracy a first-order design consideration.

Despite extensive work on localization and on compression for estimation in relative isolation (e.g.,~\cite{berger1996ceo,shen2010fundamental,shen2010fundamentalii,balkan2010crlb,wymeersch2009cooperative}),\addra{ including task-oriented quantization perspectives~\cite{Shlezinger2019TaskQuant}, a Cram\'er--Rao lower bound (CRLB) characterization for time-delay estimation (TDE) in dispersed-spectrum settings~\cite{Gezici2009DispersedSpectrumCRLB}, and CRLB--versus--rate tradeoffs for rate-constrained distributed TDE from quantized observations~\cite{Srinivasan2010RateConstrainedTDE}, }the literature that provides \emph{analytically explicit} accuracy limits for localization \emph{from compressed multi-sensor observations} remains relatively sparse, especially in wideband waveform models where time-of-arrival structure is central. In particular, while rate--distortion (RD) theory~\cite{bergerratedistortion1971} offers a tractable lens for modeling optimal compression \delra{for }under \delra{quadratic fidelity criteria}\addra{a squared-error distortion criterion}, it is unclear how RD-optimal compression for \emph{signal reconstruction} translates into localization performance, and whether compress-then-estimate design is well aligned with the information content most relevant for localization.

This letter addresses the following fundamental question: \emph{What are the accuracy limits of distributed localization when the FC has access only to\addra{ independently} \delra{rate-constrained (}compressed\delra{)} versions of the sensor observations?} Under a line-of-sight wideband Gaussian waveform model, we derive frequency-domain \delra{Cram\'er--Rao lower bounds (}CRLBs\delra{)} for localization from compressed observations by combining a spectral Fisher-information (FI) characterization with a Gaussian test-channel model of compression. We specialize the bound to a two-band, two-level spectral model, yielding closed-form expressions that reveal nontrivial rate-dependent behavior, including regime changes induced by water-filling. Finally, we construct an explicit counterexample showing that compression optimized for mean-square signal reconstruction can be strictly suboptimal for localization accuracy, which motivates localization-aware (task-oriented) compression beyond compress-then-estimate designs.

\vspace{-0.2cm}
\section{Problem Formulation}\label{sec:problem_formulation}
We consider a distributed sensing system with $M\ge 2$ spatially separated sensors.
Let $\vecp_m\in\reals^{d\times 1}$ denote the known position of sensor $m\in\{1,\ldots,M\}$,
and let $\vecp\in\reals^{d\times 1}$ denote the unknown source location, where typically $d\in\{2,3\}$.
Assuming line-of-sight propagation with known speed $c>0$, the propagation delay to sensor $m$ is
\begin{equation}\label{eq:tau_m_loc_def}
    \tau_m(\vecp)\triangleq \frac{\|\vecp-\vecp_m\|_2}{c}\in\positivereals,\quad m\in\{1,\ldots,M\}.
\end{equation}
As will be evident from the second-order statistics below, and as is standard in passive localization, the dependence on $\vecp$ is only through the time-delay differences $\Delta_{m\ell}(\vecp)\triangleq\tau_m(\vecp)-\tau_\ell(\vecp)$. Thus, no ``global time-shift''
parameter is introduced.

Each sensor observes a noisy, time-shifted version of a common wideband source over the interval $t\in[0,T]$,
\begin{equation}\label{eq:ct_obs_model_loc}
    \rndx_m(t)=\rnds\big(t-\tau_m(\vecp)\big)+\rndn_m(t)\in\reals,
    \quad m\in\{1,\ldots,M\},
\end{equation}
where $\rnds(t)$ is a zero-mean wide-sense stationary (WSS) Gaussian process with (two-sided) power spectral density (PSD) function $S_{\rnds}(f)$ supported on
$|f|\le B$ for some known $B>0$.
The noises $\{\rndn_m(t)\}_{m=1}^M$ are mutually independent, independent of $\rnds(t)$,
zero-mean WSS Gaussian processes with PSDs $\{S_{\rndn_m}(f)\}_{m=1}^M$ (supported on $|f|\le B$).

Define the sensor vector process $\rvecx(t)\triangleq\tp{[\rndx_1(t)\ \cdots\ \rndx_M(t)]}$.
Its (matrix-valued) cross-spectral density admits the form
\begin{equation}\label{eq:csd_uncompressed}
    \matS_{\rvecx}(f;\vecp)
    =
    S_{\rnds}(f)\,\vecv(f;\vecp)\herm{\vecv}(f;\vecp)
    + \matS_{\rvecn}(f)\in\complexset^{M\times M},
\end{equation}
where $\matS_{\rvecn}(f)\triangleq \mathrm{diag}\!\big(S_{\rndn_1}(f),\ldots,S_{\rndn_M}(f)\big)$ and
\begin{equation}\label{eq:v_def_loc}
    \vecv(f;\vecp)\triangleq
    \tp{\left[
    e^{-\jmath 2\pi f\tau_1(\vecp)}\,\cdots\,e^{-\jmath 2\pi f\tau_M(\vecp)}
    \right]}\in\complexset^{M\times 1}.
\end{equation}
In particular, for $m\neq \ell$,
\begin{equation}\label{eq:offdiag_tdoa}
    [\matS_{\rvecx}(f;\vecp)]_{m\ell}
    =
    S_{\rnds}(f)\,e^{-\jmath 2\pi f\Delta_{m\ell}(\vecp)}\in\complexset,
\end{equation}
which shows explicitly that $\vecp$ affects the observation law only through time-difference-of-arrival terms.

In the distributed setting with communication constraints, each sensor $m$ communicates to a FC over a rate-limited link of rate $R_m$ bits per second. Thus, the FC does not have direct access to $\rndx_m(\cdot)$, but rather to a compressed version thereof, denoted as $\widehat{\rndx}_m(\cdot)$, produced at sensor $m$ under the rate constraint $R_m$. Let $\widehat{\rvecx}(t)\triangleq\tp{[\widehat{\rndx}_1(t)\ \cdots\ \widehat{\rndx}_M(t)]}$
denote the vector process of the compressed signals available at the FC.

Based on $\widehat{\rvecx}(\cdot)$ over $[0,T]$, the FC constructs an estimator $\widehat{\vecp}$ of the unknown
location $\vecp$. The focus of this letter is to characterize the fundamental localization limits from RD-optimally compressed
observations by deriving the Cram\'er--Rao lower bound (CRLB) on the mean-square error of any unbiased estimator of $\widehat{\vecp}$ based on $\widehat{\rvecx}(\cdot)$ over $[0,T]$. In addition, and perhaps surprisingly, we will show that it is easy to design a strictly RD-suboptimal compression giving significantly higher localization accuracy, thus highlighting the need for developing joint compression-localization schemes.

\vspace{-0.2cm}
\section{CRLB for RD-Optimally Compressed Signals}\label{sec:crlb_rd}

Under the Gaussian signal model in Section~\ref{sec:problem_formulation} and the squared-error distortion measure,
the RD optimal compression of a WSS Gaussian process admits a convenient test-channel representation
(e.g.,~\cite{bergerratedistortion1971}). In particular, the $m$-th compressed waveform
available at the FC can be modeled as
\begin{equation}\label{eq:test_channel_loc_time}
    \widehat{\rndx}_m(t)=\left(b_m * \rndx_m\right)(t)+\rndz_m(t)\in\reals,
\end{equation}
where $b_m(t)$ is a linear time-invariant (LTI) filter with (real-valued) frequency response $B_m(f)$, and
$\rndz_m(t)$ is a zero-mean WSS Gaussian ``compression-noise'' process, independent of $\rndx_m(t)$, with PSD
$S_{\rndz_m}(f)$. Denoting by $S_{\rndx_m}(f)$ the PSD of $\rndx_m(t)$, the RD water-filling solution~\cite{cover1999elements} implies that there exists a
water level $\lambda_m\ge 0$ such that
\begin{align}
    B_m(f)
    = \left[1-\frac{\lambda_m}{S_{\rndx_m}(f)}\right]^+, \quad S_{\rndz_m}(f)=\lambda_m\,B_m(f),\label{eq:Bm_and_Sz_explicit}
\end{align}
where $[u]^+\triangleq\max\{u,0\}$. The water level $\lambda_m$ is uniquely determined by the rate constraint $R_m$ via
\begin{equation}\label{eq:rate_constraint_waterfilling}
    R_m =  \frac{1}{2}\int_{-\infty}^{\infty}\left[\log_2\left(\frac{S_{\rndx_m}(f)}{\lambda_m}\right)\right]^+\!{\rm d}f,
\end{equation}
where the specialized integral \eqref{eq:rate_constraint_waterfilling} in our bandlimited signal setting is effectively over $|f|\le B$.

Since each $\widehat{\rndx}_m(t)$ in \eqref{eq:test_channel_loc_time} is obtained from a linear transformation of $\rndx_m(t)$
with an addition of independent WSS Gaussian noise, the compressed vector process $\widehat{\rvecx}(t)$ remains WSS Gaussian.
Consequently, its statistical law is fully characterized by its cross-spectral density matrix
$\matS_{\widehat{\rvecx}}(f;\vecp)$, given by
\begin{equation}\label{eq:csd_compressed}
    \matS_{\widehat{\rvecx}}(f;\vecp)
    =
    S_{\rnds}(f)\,\matB(f)\vecv(f;\vecp)\herm{\vecv}(f;\vecp)\matB(f)+\matSigma(f),
\end{equation}
where we have defined
\begin{equation}\label{eq:Bmat_def}
    \matB(f)\triangleq \mathrm{diag}\!\left(B_1(f),\ldots,B_M(f)\right)\in\positivereals^{M\times M},
\end{equation}
and using $\{S_{\rndw_m}(f)\triangleq|B_m(f)|^2 S_{\rndn_m}(f)+S_{\rndz_m}(f)\}_{m=1}^M$,
\begin{equation}\label{eq:Sigma_def_loc}
    \matSigma(f)\triangleq
    \mathrm{diag}\!\left(S_{\rndw_1}(f),\ldots,S_{\rndw_M}(f)\right)\in\positivereals^{M\times M}.
\end{equation}

\vspace{-0.4cm}
\subsection{CRLB via a spectral FI}
Since $\widehat{\rvecx}(t)$ is WSS Gaussian with a cross-spectral density matrix
$\matS_{\widehat{\rvecx}}(f;\vecp)$, the localization information can be expressed in the frequency domain.
Specifically, under standard regularity conditions for purely non-deterministic stationary Gaussian processes,\footnote{For example, it is sufficient to assume that $\matS_{\widehat{\rvecx}}(f;\vecp)$ is uniformly bounded and positive definite on $[-B,B]$, and is continuously differentiable in $\vecp$.} the FI rate (i.e., per unit time) matrix admits the Whittle spectral representation~\cite{whittle1953analysis}, i.e.,
\begin{align}
    &\lim_{T\to\infty}\frac{1}{T}\left[\matJ^{\ratedist}(\vecp)\right]_{ij}\label{eq:FIMrate_whittle_loc}\\
    &=
    \frac{1}{2}\int_{-\infty}^{\infty}
    \!\tr\!\left(
    \matS_{\widehat{\rvecx}}^{-1}(f;\vecp)
    \matS'_{\widehat{\rvecx},i}(f;\vecp)
    \matS_{\widehat{\rvecx}}^{-1}(f;\vecp)
    \matS'_{\widehat{\rvecx},j}(f;\vecp)
    \right){\rm d}f,\label{eq:FIMrate_whittle_loc2}
\end{align}
where, for all $1\leq i,j\leq d$, $\left[\matJ^{\ratedist}(\vecp)\right]_{ij}$ denotes the $(i,j)$-th element of the FI rate matrix for the RD-optimally compressed signals $\widehat{\rvecx}(\cdot)$, $\matS'_{\widehat{\rvecx},i}(f;\vecp)\triangleq\frac{\partial \matS_{\widehat{\rvecx}}(f;\vecp)}{\partial p_i}$,
and in our bandlimited model the integral reduces to $f\in[-B,B]$.
For an observation horizon $T$, the CRLB is (e.g.,~\cite{van2004detection})
\begin{equation}\label{eq:crlb_loc_T}
    \Covop\left(\widehat{\vecp}-\vecp\right)\triangleq \Exop\left[  \left(\widehat{\vecp}-\vecp\right)\tp{\left(\widehat{\vecp}-\vecp\right)} \right]\succeq {\matJ_T^{\ratedist}}^{-1}(\vecp),
\end{equation}
where $\matJ^{\ratedist}_T(\vecp)$ denotes the FI matrix (FIM) for $\vecp$ based on $\widehat{\rvecx}(\cdot)$ over $[0,T]$. Moreover, $\matJ^{\ratedist}_T(\vecp)$ grows linearly with $T$, and
\begin{equation}\label{eq:Jrate_def}
    \matJ^{\ratedist}_\infty(\vecp)\triangleq \lim_{T\to\infty}\frac{1}{T}\matJ^{\ratedist}_T(\vecp)
\end{equation}
exists and is given by \eqref{eq:FIMrate_whittle_loc}. Thus, $\matJ^{\ratedist}_T(\vecp)=T\,\matJ^{\ratedist}_\infty(\vecp)+o(T)$.

To evaluate \eqref{eq:FIMrate_whittle_loc2}, it remains to compute
$\matS'_{\widehat{\rvecx},i}(f;\vecp)$. Recall \eqref{eq:csd_compressed}, and that $\matSigma(f)$ is independent of $\vecp$. Thus, the dependence in $\vecp$ is only through the phase terms
$e^{-\jmath 2\pi f\tau_m(\vecp)}$ in $\vecv(f;\vecp)$. Differentiating \eqref{eq:tau_m_loc_def} yields, for all $1\leq i\leq d$,
\begin{equation}\label{eq:delay_grad_def}
    \frac{\partial \tau_m(\vecp)}{\partial p_i}
    =
    \frac{1}{c}\,\frac{p_i-[\vecp_m]_i}{\|\vecp-\vecp_m\|_2},
    \quad m\in\{1,\ldots,M\},
\end{equation}
and consequently, for each $m$,
\begin{equation}\label{eq:dv_dp}
    \frac{\partial [\vecv(f;\vecp)]_m}{\partial p_i}
    =
    -\jmath 2\pi f\frac{\partial \tau_m(\vecp)}{\partial p_i}\,[\vecv(f;\vecp)]_m.
\end{equation}
Using the notation $\vecv'_{i}(f;\vecp)\triangleq\frac{\partial \vecv(f;\vecp)}{\partial p_i}$, by the product rule,
\begin{align}\label{eq:dS_dp_general}
    \matS'_{\widehat{\rvecx},i}(f;\vecp)&=S_{\rnds}(f)\,\matB(f)\matQ_i(f;\vecp)
    \matB(f),
\end{align}
where $\matQ_i(f;\vecp)\triangleq\vecv'_{i}(f;\vecp)\herm{\vecv}(f;\vecp)+\vecv(f;\vecp)\herm{\vecv'}_{i}(f;\vecp)$, and which together with \eqref{eq:FIMrate_whittle_loc} yields an explicit CRLB expression as a one-dimensional integral over $f\in[-B,B]$.
We next specialize this expression to a two-band, two-level spectral model to obtain fully closed-form
bounds and to show, with a simple RD-suboptimal compression, that RD-optimal compression for reconstruction can be strictly suboptimal for localization.

\vspace{-0.3cm}
\section{RD Is Not Localization-Optimal: A Two-Band Counterexample}\label{sec:counterexample}
The CRLB expression derived in Section~\ref{sec:crlb_rd} is explicit, but still involves a frequency integral over the
(possibly arbitrary) source PSD $S_{\rnds}(f)$, and the RD water-filling induces rate-dependent changes in the
effective spectrum seen at the FC. It is therefore generally unclear whether RD-optimality for waveform reconstruction transfers to optimality in terms of localization accuracy (or even preserves the most localization-informative spectral components). In this section, we adopt an explicit two-band, two-level spectral model that
(i) yields closed-form expressions for the CRLB; and (ii) already captures the key phenomenon: RD-optimal compression for
signal reconstruction can be strictly suboptimal for localization.

\vspace{-0.3cm}
\subsection{Two-band, two-level spectral model}\label{subsec:two_band_model}
Fix $0<f_{\rm L}<f_{\rm H}\le B$ and define the two disjoint bands
\begin{equation}\label{eq:two_bands_def}
    \setB_{\rm L}\triangleq\{f:|f|\le f_{\rm L}\},
    \qquad
    \setB_{\rm H}\triangleq\{f:f_{\rm L}<|f|\le f_{\rm H}\}.
\end{equation}
We assume the following piecewise-constant source spectrum
\begin{equation}\label{eq:two_level_psd}
    S_{\rnds}(f)=
    \begin{cases}
        S_{\rm L}, & f\in \setB_{\rm L},\\
        S_{\rm H}, & f\in \setB_{\rm H},\\
        0, & |f|>f_{\rm H},
    \end{cases}
\end{equation}
with $S_{\rm L}>S_{\rm H}>0$. For simplicity, we also assume identical sensor noises with flat PSD in the relevant band,
\begin{equation}\label{eq:flat_noise_psd_counter}
    S_{\rndn_m}(f)=
    \begin{cases}
        N_0, & |f|\le f_{\rm H},\\
        0, & |f|>f_{\rm H},
    \end{cases}
    \quad \forall m\in\{1,\ldots,M\},
\end{equation}
so that $S_{\rndx_m}(f)=S_{\rnds}(f)+N_0$ for $|f|\le f_{\rm H}$.

Under \eqref{eq:two_level_psd}--\eqref{eq:flat_noise_psd_counter}, the RD water-filling solution
\eqref{eq:Bm_and_Sz_explicit} is also piecewise constant over the two bands. In particular, letting
$\lambda_m$ denote the water level at sensor $m$, we have
\begin{equation}\label{eq:Bm_two_band}
    B_m(f)=
    \begin{cases}
        \left[1-\frac{\lambda_m}{S_{\rm L}+N_0}\right]^+, & f\in\setB_{\rm L},\\[4pt]
        \left[1-\frac{\lambda_m}{S_{\rm H}+N_0}\right]^+, & f\in\setB_{\rm H},\\[4pt]
        0, & |f|>f_{\rm H},
    \end{cases}
\end{equation}
and $S_{\rndz_m}(f)=\lambda_m B_m(f)$ for $|f|\le f_{\rm H}$. As $R_m$ decreases, $\lambda_m$ increases and a \emph{regime change} occurs at $\lambda_m=S_{\rm H}+N_0$:
\begin{itemize}[itemsep=1pt,topsep=2pt,leftmargin=*]
    \item \textbf{High compression rates}: when $\lambda_m< S_{\rm H}+N_0$, both bands are active and $B_m(f)>0$ on $\setB_{\rm L}\cup\setB_{\rm H}$.
    \item \textbf{Intermediate compression rates}: when $S_{\rm H}+N_0\le \lambda_m < S_{\rm L}+N_0$, then $B_m(f)>0$ on $\setB_{\rm L}$ but $B_m(f)=0$ on $\setB_{\rm H}$. Critically, the RD solution \emph{drops the high-frequency band}.
    \item \textbf{Low compression rates}: when $\lambda_m\ge S_{\rm L}+N_0$, $B_m(f)\equiv 0$ and no information is conveyed.
\end{itemize}
Thus, in this two-band two-level spectral model, RD-optimal compression can eliminate an entire band as the rate decreases.

However, recall from \eqref{eq:dv_dp} that the geometry dependence enters through the phase terms
$e^{-\jmath 2\pi f\tau_m(\vecp)}$, and differentiation introduces a factor of $f$.
Consequently, the integrand in \eqref{eq:FIMrate_whittle_loc2} contains an intrinsic \emph{frequency-squared} weighting:
each derivative matrix $\matS'_{\widehat{\rvecx},i}(f;\vecp)$ is proportional to $f$, hence the trace term
$\tr\!\big(\matS_{\widehat{\rvecx}}^{-1}\matS'_{\widehat{\rvecx},i}\matS_{\widehat{\rvecx}}^{-1}\matS'_{\widehat{\rvecx},j}\big)$ scales as $f^2$ (up to other spectral factors). Intuitively, high-frequency components can therefore be substantially more informative for localization, even when they carry less signal energy.

\vspace{-0.5cm}
\subsection{Closed-form CRLB under the two-band model}\label{subsec:crlb_two_band_closed}
We now specialize \eqref{eq:FIMrate_whittle_loc2} under \eqref{eq:two_level_psd}--\eqref{eq:Bm_two_band} and derive a closed-form expression. For clarity of exposition (and since the phenomenon is per-sensor), we assume the symmetric setting
\begin{equation}\label{eq:symmetric_rates}
    R_m=R \; \Rightarrow \begin{cases}
        \lambda_m=\lambda,\\
        B_m(f)=B(f),
    \end{cases} \forall m\in\{1,\ldots,M\}.
\end{equation}

For $b\in\{{\rm L},{\rm H}\}$, define the per-band effective SNR level
\begin{equation}\label{eq:gamma_b_def_thm}
    \gamma_b(R)\triangleq \frac{S_b\,B_b^2(R)}{S_{\rndw,b}(R)},
    \quad
    S_{\rndw,b}(R)\triangleq B_b^2(R)N_0+\lambda(R)\,B_b(R),
\end{equation}
where $S_{\rm L},S_{\rm H}$ are given in \eqref{eq:two_level_psd}, $\lambda(R)$ is the RD water level satisfying \eqref{eq:rate_constraint_waterfilling}, and $B_b(R)\triangleq B(f)$ for any $f\in\setB_b$ (which is constant over each band under \eqref{eq:Bm_two_band}).
Define the band information weights,
\begin{equation}\label{eq:band_weights_def_thm}
    w_b(R)\triangleq \frac{2M\gamma_b(R)^2}{1+M\gamma_b(R)},
    \qquad b\in\{{\rm L},{\rm H}\},
\end{equation}
and with it the lower and higher frequency information terms,
\begin{equation}\label{eq:lowandhighfreqinfo}
    J_{\rm L}(R)\triangleq f_{\rm L}^3w_{\rm L}(R), \quad J_{\rm H}(R)\triangleq \left(f_{\rm H}^3-f_{\rm L}^3\right)w_{\rm H}(R).
\end{equation}
Finally, define (entrywise) the geometry matrix $\matG(\vecp)\in\reals^{d\times d}$,
\begin{equation}\label{eq:G_def_thm}
    [\matG(\vecp)]_{ij}\triangleq \tp{\widetilde{\vecg}}_i(\vecp)\widetilde{\vecg}_j(\vecp),
    \qquad
    \widetilde{\vecg}_i(\vecp)\triangleq \matP\vecg_i(\vecp),
\end{equation}
where $\matP\triangleq \matI_M-\frac{1}{M}\vecone\tp{\vecone}\in\reals^{M\times M}$ is a projection matrix and
\begin{align}
    \vecg_i(\vecp)
    &\triangleq
    \tp{\left[
    \frac{\partial \tau_1(\vecp)}{\partial p_i}\ \cdots\
    \frac{\partial \tau_M(\vecp)}{\partial p_i}
    \right]}\in\reals^{M\times 1}.\label{eq:gi_app_def1_new}
\end{align}

We are now ready to state our main result.
\begin{theorem}[Closed-form FI rate and CRLB under the two-band model]\label{thm:crlb_two_band_closed}
Under the two-band, two-level spectrum model \eqref{eq:two_level_psd}--\eqref{eq:flat_noise_psd_counter}, RD-optimal compression \eqref{eq:Bm_two_band} and the symmetric setting \eqref{eq:symmetric_rates}, the FI rate for $\vecp$ based on $\widehat{\rvecx}(\cdot)$ over $[0,T]$ reads
\begin{equation}\label{eq:Jinf_two_band}
    \matJ^{\ratedist}_\infty(\vecp)
    =
    \frac{4\pi^2}{3}
    \left(J_{\rm L}(R)+J_{\rm H}(R)\right)\matG(\vecp).
\end{equation}
Consequently, for any unbiased estimator $\widehat{\vecp}$ of $\vecp$ based on $\widehat{\rvecx}(\cdot)$ over $[0,T]$, it holds that
\begin{align}
    \Covop\!\left(\widehat{\vecp}-\vecp\right)&\succeq {\matJ^{\ratedist}}_T^{-1}(\vecp)
    =
    \frac{1}{T}\,{\matJ^{\ratedist}}_\infty^{-1}(\vecp)+o\!\left(\frac{1}{T}\right)\\
    &=\frac{3}{4\pi^2 T}
    \frac{1}{J_{\rm L}(R)+J_{\rm H}(R)}
    \,\matG^{-1}(\vecp)+o\!\left(\frac{1}{T}\right),\label{eq:crlb_two_band}
\end{align}
whenever $\matG(\vecp)$ is nonsingular.
\end{theorem}

\begin{IEEEproof}
See Appendix~\ref{app:geometry_derivation}.
\end{IEEEproof}
\begin{remark}
Note that the dependence on $\vecp$ is entirely through the geometry-dependent centered delay gradients $\widetilde{\vecg}_i=\matP\vecg_i$.
\end{remark}

Under \eqref{eq:two_level_psd}--\eqref{eq:flat_noise_psd_counter}, the PSD of $\rndx_m(\cdot)$ is $S_{\rndx_m}(f)=S_{\rnds}(f)+N_0$ for $|f|\le f_{\rm H}$, hence the rate constraint \eqref{eq:rate_constraint_waterfilling} becomes
\begin{equation}\label{eq:R_lambda_two_band}
    R
    =
    f_{\rm L}\left[\log_2\!\left(\tfrac{S_{\rm L}+N_0}{\lambda}\right)\right]^+
    +(f_{\rm H}-f_{\rm L})\left[\log_2\!\left(\tfrac{S_{\rm H}+N_0}{\lambda}\right)\right]^+\!.
\end{equation}
Define the \emph{critical rate} (corresponding to $\lambda=S_{\rm H}+N_0$),
\begin{equation}\label{eq:Rcrit_two_band}
    R_{\rm crit}\triangleq f_{\rm L}\log_2\!\left(\frac{S_{\rm L}+N_0}{S_{\rm H}+N_0}\right).
\end{equation}
Then $\lambda(R)$ is explicit:
\begin{itemize}[itemsep=1pt,topsep=2pt,leftmargin=*]
\item \textbf{High-rates} ($R>R_{\rm crit}$): both bands are active ($\lambda<S_{\rm H}+N_0$) and
\begin{equation}\label{eq:lambda_high_two_band}
    \lambda(R)=
    2^{-\frac{R}{f_{\rm H}}}\,
    (S_{\rm L}+N_0)^{\frac{f_{\rm L}}{f_{\rm H}}}\,
    (S_{\rm H}+N_0)^{\frac{f_{\rm H}-f_{\rm L}}{f_{\rm H}}}.
\end{equation}
Thus $B_b(R)=1-\lambda(R)/(S_b+N_0)$ for $b\in\{{\rm L},{\rm H}\}$, and $w_b(R)$ follows from
\eqref{eq:gamma_b_def_thm}--\eqref{eq:band_weights_def_thm}.
\item \textbf{Intermediate-rates} ($0<R\le R_{\rm crit}$): only $\setB_{\rm L}$ is active ($S_{\rm H}+N_0\le\lambda<S_{\rm L}+N_0$) and
\begin{align}
    \lambda(R)&=(S_{\rm L}+N_0)\,2^{-R/f_{\rm L}},\\
    B_{\rm L}(R)&=1-2^{-R/f_{\rm L}},
    \quad
    B_{\rm H}(R)=0,\label{eq:lambda_mid_two_band}
\end{align}
so that $w_{\rm H}(R)=0$ and accordingly $J_{\rm H}(R)=0$.
\end{itemize}

Along with \eqref{eq:Rcrit_two_band}--\eqref{eq:lambda_mid_two_band}, \eqref{eq:crlb_two_band} provides an explicit closed-form CRLB parameterized by the physically meaningful quantities $(S_{\rm L},S_{\rm H},N_0,f_{\rm L},f_{\rm H},M,R,T,\{\vecp_m\},c)$. Specifically, when $R$ drops below $R_{\rm crit}$, RD-optimal compression for signal reconstruction eliminates $\setB_{\rm H}$ (hence $w_{\rm H}(R)=0$), even though the Fisher integrand scales as $f^2$, indicating that higher-frequency components can carry more localization-related information. This mechanism is generally localization-suboptimal and can lead to a catastrophic degradation in localization accuracy.

To see this more clearly, consider the following example. Fix a rate $R\le R_{\rm crit}$, for which RD-optimal compression satisfies $B_{\rm H}(R)=0$ and hence $J_{\rm H}(R)=0$.
Now, consider instead the following (RD-suboptimal) \emph{band-selective} compression scheme:
each sensor suppresses $\setB_{\rm L}$ and applies the Gaussian test channel only on $\setB_{\rm H}$, i.e.,
\begin{equation}\label{eq:band_selective_B}
    B^{\selscheme}(f)=
    \begin{cases}
        0, & f\in\setB_{\rm L},\\
        \left[1-\frac{\lambda_{\selscheme}}{S_{\rm H}+N_0}\right]^+, & f\in\setB_{\rm H},
    \end{cases}
\end{equation}
such that $S_{\rndz}^{\selscheme}(f)=\lambda_{\selscheme} B^{\selscheme}(f)\ \text{ on }\ \setB_{\rm H}$ and zero otherwise.
Imposing the \emph{same rate} $R$ over bandwidth $(f_{\rm H}-f_{\rm L})$ gives
\begin{align}\label{eq:lambda_sel}
    R&=(f_{\rm H}-f_{\rm L})\left[\log_2\!\left(\frac{S_{\rm H}+N_0}{\lambda_{\selscheme}}\right)\right]^+,\\
    \lambda_{\selscheme}(R)&=(S_{\rm H}+N_0)\,2^{-R/(f_{\rm H}-f_{\rm L})}.
\end{align}
Hence on $\setB_{\rm H}$,
\begin{equation}\label{eq:Bsel_H}
    B_{\rm H}^{\selscheme}(R),
    \qquad
    B_{\rm L}^{\selscheme}(R)=0.
\end{equation}

Define $\gamma_{\rm H}^{\selscheme}(R)$ and $w_{\rm H}^{\selscheme}(R)$ by the same formulas as in
\eqref{eq:gamma_b_def_thm}--\eqref{eq:band_weights_def_thm}, but with $(B_b,\lambda)$ replaced by $(B_b^{\selscheme},\lambda_{\selscheme})$.
Then, the FI rate under the band-selective scheme satisfies
\begin{equation}\label{eq:Jinf_sel}
    \matJ_{\infty}^{\selscheme}(\vecp)
    =
    \frac{4\pi^2}{3}\,J_{\rm H}^{\selscheme}(R)\,\matG(\vecp),
    \quad
    J_{\rm H}^{\selscheme}(R)\triangleq (f_{\rm H}^3-f_{\rm L}^3)\,w_{\rm H}^{\selscheme}(R),
\end{equation}
while under RD (for $R\le R_{\rm crit}$) we have $J_{\rm H}^{\ratedist}(R)=0$, hence
\begin{equation}\label{eq:Jinf_RD_lowrate}
    \matJ_{\infty}^{\ratedist}(\vecp)
    =
    \frac{4\pi^2}{3}\,J_{\rm L}^{\ratedist}(R)\,\matG(\vecp),
    \quad
    J_{\rm L}^{\ratedist}(R)=f_{\rm L}^3\,w_{\rm L}(R).
\end{equation}
If $J_{\rm H}^{\selscheme}(R)\!>\!J_{\rm L}^{\ratedist}(R)$, then $\matJ_{\infty}^{\selscheme}(\vecp)\!\succ\! \matJ_{\infty}^{\ratedist}(\vecp)$ and, consequently,
\begin{equation}\label{eq:crlb_improvement_diag}
    \frac{\big[\matJ_T^{-1}(\vecp)\big]_{ii}^{\selscheme}}{\big[\matJ_T^{-1}(\vecp)\big]_{ii}^{\ratedist}}
    =
    \frac{J_{\rm L}^{\ratedist}(R)}{J_{\rm H}^{\selscheme}(R)}<1,\quad 1\leq i\leq d,
\end{equation}
i.e., the CRLB for the band-selective scheme is strictly smaller. 

As a more concrete example, consider a mmWave/ISAC-like wideband regime~\cite{gao2022integrated,guo2025integrated} with $M=4$ sensors, wherein $f_{\rm H}=200~{\rm MHz}$, and further fix $f_{\rm L}=5~{\rm MHz}$, $S_{\rm L}=100$, $S_{\rm H}=20$, and $N_0=1$. For these values, $R_{\rm crit}\approx 11.3~{\rm Mb/s}$. For a rate $R=10~{\rm Mb/s}<R_{\rm crit}$ (a plausible per-sensor backhaul budget), using the expressions from Theorem~\ref{thm:crlb_two_band_closed}, we obtain in this setting
\begin{equation}\label{eq:num_example_ratio_select_realistic}
    \frac{\big[\matJ_T^{-1}(\vecp)\big]_{ii}^{\selscheme}}{\big[\matJ_T^{-1}(\vecp)\big]_{ii}^{\ratedist}}
    =
    \frac{J_{\rm L}^{\ratedist}(R)}{J_{\rm H}^{\selscheme}(R)}
    \approx 9.97\times10^{-3} ,\quad 1\le i\le d,
\end{equation}
i.e., at the \emph{same} rate $R$ the band-selective scheme yields a reduction of two orders of magnitude in the CRLB.

\vspace{-0.1cm}
\section{Discussion and Outlook}\label{sec:conclusion}
This letter derived fundamental localization limits when a FC has access only to rate-constrained
(compressed) versions of wideband multi-sensor observations. Leveraging a Gaussian line-of-sight waveform model and the
Gaussian RD test-channel representation, we obtained an explicit frequency-domain characterization of the Fisher
information and the associated CRLB. Specializing further to a simple, yet insightful two-band, two-level spectrum model yielded closed-form
expressions that transparently expose the interaction between rate allocation, spectral content,
and localization accuracy.

Beyond providing an analytically tractable CRLB, the two-band specialization highlights a key conceptual message:\addra{ under
the Gaussian per-sensor RD benchmark,} compression optimized for waveform reconstruction can be poorly aligned with localization. In particular, under RD-optimal compression, decreasing the rate can induce a sharp regime change in which high-frequency components are eliminated when the compression rate is sufficiently low, despite the fact that localization information is inherently
weighted towards higher frequencies due to the phase sensitivity.
This reveals an explicit mechanism by which compress-then-estimate designs\addra{ based on reconstruction-optimal per-sensor RD
compression} may incur a dramatic loss in localization accuracy\addra{.}\delra{, thereby motivating}\addra{ This observation, in turn, motivates} task- or goal-oriented compression strategies that preserve the most localization-informative signal features.

Naturally, there are many related important questions that remain to be addressed. While our analysis focused on RD-optimal compression under the quadratic distortion criterion, a central direction is to characterize and design \emph{localization-aware} compression rules, e.g., by formulating
rate-allocation problems that maximize localization fidelity subject to rate constraints. It is also
of\addra{ both theoretical and practical} interest to extend the framework beyond line-of-sight to multipath and cluttered environments, and to integrate the
resulting bounds into ISAC-oriented system design, including bandwidth allocation, and
joint sensing--communication resource management.

\appendices
\setcounter{equation}{0}
\renewcommand{\theequation}{S\arabic{equation}}

\section{Derivation of the CRLB for the two-band model}\label{app:geometry_derivation}

Fix a band $b\in\{{\rm L},{\rm H}\}$ and a frequency $f\in\setB_b$. Under
\eqref{eq:two_level_psd}--\eqref{eq:flat_noise_psd_counter} and the symmetric assumption \eqref{eq:symmetric_rates},
the compressed cross-spectral density \eqref{eq:csd_compressed} reduces on $\setB_b$ to
\begin{align}
    \matS_{\widehat{\rvecx}}(f;\vecp)
    &=
    S_b B_b^2\,\vecv(f;\vecp)\herm{\vecv}(f;\vecp)
    + S_{\rndw,b}\matI_M,\label{eq:Phi_app_1_new}
\end{align}
where $S_b\in\{S_{\rm L},S_{\rm H}\}$ is the source PSD level on $\setB_b$, $B_b\triangleq B(f)$ is constant on $\setB_b$, and recall \eqref{eq:gamma_b_def_thm}, so that
\begin{align}
    \matS_{\widehat{\rvecx}}(f;\vecp)
    =
    S_{\rndw,b}\Big(\matI_M + \gamma_b\,\vecv(f;\vecp)\herm{\vecv}(f;\vecp)\Big).\label{eq:Phi_app_2_new}
\end{align}
Since $\herm{\vecv}(f;\vecp)\vecv(f;\vecp)\!=\!M$, the Sherman--Morrison formula gives
\begin{align}
    \matS_{\widehat{\rvecx}}^{-1}(f;\vecp)
    &=
    \tfrac{1}{S_{\rndw,b}}\matK_b(f;\vecp),\label{eq:Phi_inv_app_new}\\
    \matK_b(f;\vecp)
    &\triangleq
    \matI_M-\tfrac{\gamma_b}{1+M\gamma_b}\,\vecv(f;\vecp)\herm{\vecv}(f;\vecp).\label{eq:Kb_app_def_new}
\end{align}

Next, recall \eqref{eq:delay_grad_def} and \eqref{eq:gi_app_def1_new}, and with it define the diagonal matrix $\matD_i(\vecp)\triangleq \mathrm{diag}\!\left(\vecg_i(\vecp)\right)$. Then, \eqref{eq:dv_dp} becomes
\begin{align}
    \vecv'_i(f;\vecp)=-\jmath 2\pi f\,\matD_i(\vecp)\vecv(f;\vecp).\label{eq:vi_app_new}
\end{align}
From \eqref{eq:symmetric_rates}, $\matB(f)=B_b\matI_M$ on $\setB_b$, hence \eqref{eq:dS_dp_general} gives
\begin{align}
    \matS'_{\widehat{\rvecx},i}(f;\vecp)
    &=
    S_b B_b^2\left(\vecv'_i\herm{\vecv}+\vecv\herm{\vecv'}_i\right)\\
    &=
    \jmath 2\pi f\,S_b B_b^2\left(\vecv\herm{\vecv}\matD_i-\matD_i\vecv\herm{\vecv}\right),\; f\in\setB_b,\label{eq:Sprime_app_new}
\end{align}
where we suppress $(f;\vecp)$ and $(\vecp)$ for brevity.

We recall that, with the projection matrix $\matP=\matI_M-\frac{1}{M}\vecone\tp{\vecone}$, we have $\widetilde{\vecg}_i(\vecp)= \matP\vecg_i(\vecp)$~\eqref{eq:G_def_thm}, and by further defining $\bar g_i(\vecp)\triangleq\frac{1}{M}\vecone\tp\vecg_i(\vecp)$, we have $\widetilde{\vecg}_i(\vecp)=\vecg_i(\vecp)-\bar g_i(\vecp)\vecone$. Note also that $\herm{\vecv}\matD_i\vecv=\sum_{m=1}^M g_{m,i}=M\bar g_i$ since $|[\vecv]_m|=1$.

To prove the theorem, we will use the following key lemma.
\begin{lemma}\label{lem:trace_identity}
For each $b\in\{{\rm L},{\rm H}\}$ and $f\in\setB_b$,
\begin{equation}
\begin{aligned}\label{eq:trace_identity_app}
&\tr\!\Big(
\matS_{\widehat{\rvecx}}^{-1}(f;\vecp)\matS'_{\widehat{\rvecx},i}(f;\vecp)
\matS_{\widehat{\rvecx}}^{-1}(f;\vecp)\matS'_{\widehat{\rvecx},j}(f;\vecp)
\Big)
\\
&=
(2\pi f)^2\,w_b(R)[\matG(\vecp)]_{ij}.
\end{aligned}
\end{equation}
\end{lemma}

\begin{IEEEproof}
Fix $b$ and $f\in\setB_b$, and suppress $(f;\vecp)$ for brevity. Using \eqref{eq:Phi_inv_app_new} and \eqref{eq:Sprime_app_new},
\begin{align}
\tr\!\Big(\matS^{-1}\matS'_i\matS^{-1}\matS'_j\Big)
&=
\frac{\big(\jmath 2\pi f\,S_b B_b^2\big)^2}{S_{\rndw,b}^2}
\tr\!\Big(\matK_b\matA_i\matK_b\matA_j\Big)\notag\\
&=
-(2\pi f)^2\gamma_b^2\,\tr\!\Big(\matK_b\matA_i\matK_b\matA_j\Big),
\label{eq:trace_start_new}
\end{align}
where $\matA_i\triangleq \matD_i\vecv\herm{\vecv}-\vecv\herm{\vecv}\matD_i$ and we used $\jmath^2=-1$.

Define $\widetilde{\matD}_i\triangleq \matD_i-\bar g_i\matI_M$ and $\vecu_i\triangleq \widetilde{\matD}_i\vecv$.
Then $\herm{\vecv}\vecu_i=\herm{\vecv}\widetilde{\matD}_i\vecv=M\bar g_i-M\bar g_i=0$, hence $\vecu_i\perp \vecv$, and
\begin{align}
    \matA_i
    =
    \vecu_i\herm{\vecv}-\vecv\herm{\vecu}_i.\label{eq:A_as_uv_new}
\end{align}
Moreover, since $\vecu_i=\mathrm{diag}(\widetilde{\vecg}_i)\vecv$ and $|[\vecv]_m|=1$,
\begin{align}
    \herm{\vecu}_i\vecu_j
    =
    \sum_{m=1}^M \widetilde g_{m,i}\widetilde g_{m,j}
    =
    \tp{\widetilde{\vecg}}_i\widetilde{\vecg}_j=[\matG(\vecp)]_{ij},\label{eq:uu_equals_g_new}
\end{align}
where we recall the definition \eqref{eq:G_def_thm} of the geometry matrix.

Next, note that $\matK_b\vecv=\frac{1}{1+M\gamma_b}\vecv$ and $\matK_b\vecu_i=\vecu_i$ (since $\vecu_i\perp \vecv$ and
$\matK_b$ is a rank-one modification along $\mathrm{span}\{\vecv\}$). Therefore, by \eqref{eq:A_as_uv_new} we have $\matK_b\matA_i\matK_b=\tfrac{1}{1+M\gamma_b}\matA_i$,
hence
\begin{align}
    \tr\!\Big(\matK_b\matA_i\matK_b\matA_j\Big)
    =
    \tfrac{1}{1+M\gamma_b}\tr(\matA_i\matA_j).\label{eq:reduce_to_AAj_new}
\end{align}
Finally, using \eqref{eq:A_as_uv_new} and $\herm{\vecv}\vecu_i=\herm{\vecu}_i\vecv=0$,
\begin{align}
    \matA_i\matA_j
    &=
    (\vecu_i\herm{\vecv}-\vecv\herm{\vecu}_i)(\vecu_j\herm{\vecv}-\vecv\herm{\vecu}_j)\\
    &=
    -\vecu_i(\herm{\vecv}\vecv)\herm{\vecu}_j
    -\vecv(\herm{\vecu}_i\vecu_j)\herm{\vecv},
\end{align}
so that
\begin{align}
    \tr(\matA_i\matA_j)
    &=
    -M\,\tr(\vecu_i\herm{\vecu}_j)-(\herm{\vecu}_i\vecu_j)\tr(\vecv\herm{\vecv})\notag\\
    &=
    -M(\herm{\vecu}_j\vecu_i)-M(\herm{\vecu}_i\vecu_j)
    =
    -2M\,\herm{\vecu}_i\vecu_j.\label{eq:trace_AAj_new}
\end{align}
Combining \eqref{eq:trace_start_new}, \eqref{eq:uu_equals_g_new}, \eqref{eq:reduce_to_AAj_new} and \eqref{eq:trace_AAj_new} yields
\begin{align}
    \tr\!\Big(\matS^{-1}\matS'_i\matS^{-1}\matS'_j\Big)
    &=
    (2\pi f)^2\gamma_b^2\frac{2M}{1+M\gamma_b}[\matG(\vecp)]_{ij},
\end{align}
which proves the lemma.
\end{IEEEproof}

Using Lemma~\ref{lem:trace_identity} in \eqref{eq:FIMrate_whittle_loc2} and summing the two bands gives
\begin{align}
    \left[\matJ^{\ratedist}_\infty(\vecp)\right]_{ij}
    &\!=\!
    2\pi^2\!
    \left(\sum_{b\in\{{\rm L},{\rm H}\}}\int_{\setB_b} f^2\,w_b(R)\,{\rm d}f\right)
    [\matG(\vecp)]_{ij}.\label{eq:Jrate_app_ij}
\end{align}
or, equivalently, in matrix form,
\begin{align}
    \matJ^{\ratedist}_\infty(\vecp)
    \!=\!
    2\pi^2
    \!\left(\!\sum_{b\in\{{\rm L},{\rm H}\}}\int_{\setB_b} f^2\,w_b(R)\,{\rm d}f\right)\matG(\vecp).\label{eq:Jrate_app_matrix}
\end{align}
Since the scalar term is constant over each band, we obtain
\begin{align}
    \matJ^{\ratedist}_\infty(\vecp)
    =
    \frac{4}{3}\pi^2
    \left( J_{\rm L} + J_{\rm H}\right)\matG(\vecp),\label{eq:Jrate_app_closed}
\end{align}
where lower and higher frequency information terms, $J_{\rm L}$ and $J_{\rm H}$, respectively, are defined in \eqref{eq:lowandhighfreqinfo}. Using the relation \eqref{eq:Jrate_def} and inverting the FIM $\matJ^{\ratedist}_T(\vecp)$ concludes the proof.

\bibliographystyle{IEEEbib}
\bibliography{Inputs/refs}

\newpage

\end{document}